\documentclass[conference]{IEEEtran}
\usepackage{amsmath,cite,amsfonts,amssymb,psfrag}
\usepackage{graphicx}
\usepackage{rotating}
\usepackage{color}

\newtheorem{theorem}{Theorem}

\newtheorem{remark}{Remark}
\newtheorem{definition}{Definition}

\newtheorem{lemma}{Lemma}
\newtheorem{corollary}{Corollary}
\newtheorem{proof}{Proof}
\newcounter{MYtempeqncnt}
\begin{document}
\title{Privacy-Constrained Remote Source Coding}
\IEEEoverridecommandlockouts
\author{
\authorblockN{Kittipong Kittichokechai and Giuseppe Caire\\
\authorblockA{Technische Universit\"{a}t Berlin}
}
}
\maketitle
\begin{abstract}
 We consider the problem of revealing/sharing data in an efficient and secure way via a compact representation. The representation should ensure reliable reconstruction of the desired features/attributes while still preserve privacy of the secret parts of the data. The problem is formulated as a remote lossy source coding with a privacy constraint where the remote source consists of public and secret parts.  Inner and outer bounds for the optimal tradeoff region of compression rate, distortion, and privacy leakage rate are given and shown to coincide for some special cases. When specializing the distortion measure to a logarithmic loss function, the resulting rate-distortion-leakage tradeoff for the case of identical side information forms an optimization problem which corresponds to the ``secure" version of the so-called {\em information bottleneck}.
\end{abstract}
\section{Introduction}\label{sec:introduction}
With the prominence of the Internet and the rise of the Internet of Things (IoT), significant amount of data are being generated, stored, and exchanged over the networks. Proper data management has become one of the most important and challenging aspects in system design. Information contained in the data are usually valuable resources that can be harnessed.
However, the extensive use of data incurs some privacy risks especially when sensitive information is involved. 
The ultimate goal is to utilize the data to its full extent while still preserving  privacy of the sensitive information.

Sankar et al. \cite{srpUPTI13} and du Pin Calmon and Fawaz \cite{dfPASI12}  studied the utility-privacy tradeoff from an information theoretic perspective, by relating to the framework of secure lossy source coding \cite{yASCP83,yCTFS94,prOSDS07,gepLCWS08,vpSMSC13,scRDTS14}. Inspired by these works, we consider a problem of secure \emph{remote} source coding where the remote source consists of \emph{public} and \emph{secret} parts (hidden information associated with the data). The legitimate receiver and eavesdropper are assumed to have access to the compact representation of the data as well as separate side information. The goal is to extract public attribute/feature of the data at the legitimate receiver from the compact representation satisfying a distortion criterion, while ensuring a low amount of information leakage of secret attribute/feature of the data to the eavesdropper. Similarly as in \cite{srpUPTI13}, we capture utility of the data by the reconstruction distortion of the public part at the legitimate receiver, and capture the privacy leakage by the normalized mutual information between the secret part and the eavesdropper observation. The compression rate is also considered as a practical constraint on limited storage. We wish to characterize the optimal tradeoff region of the compression rate, incurred distortion, and  privacy leakage rate. In this work, inner and outer bounds to the optimal tradeoff are given and shown to be tight for some special cases. The results can be relevant for data sharing scenarios  with (external) attacks on sensitive information. 

In the same spirit as how the the rate-distortion theorem under logarithmic loss  \cite{mfUP98,cwMSCU14,msfmFIBT14} is related to the information bottleneck problem \cite{tpbTIBM00}, in this work, the rate-distortion-privacy leakage tradeoff under logarithmic loss distortion corresponds to the secure version of information bottleneck.
It extends the information bottleneck by including the privacy constraint or alternatively it extends the \emph{dual} privacy funnel problem \cite{msfmFIBT14} by including the compression rate constraint. The secure remote source coding problem under logarithmic loss distortion therefore gives an operational meaning to the secure information bottleneck problem. 

\begin{figure}[]
	\centering
	\psfrag{x}[][][0.9]{\small{$X^{n}$}}
	\psfrag{y}[][][0.9]{\small{$Y^n$}}
	\psfrag{z}[][[0.9]{\small{$Z^n$}}
	\psfrag{yps}[][][0.9]{\small{$(Y_p^n,Y_s^n)$}}
	\psfrag{yha}[][][0.9]{\small{$\hat{Y}_p^n$}}
	\psfrag{d}[][][0.9]{\small{$E[d(Y_p^n,\hat{Y}_p^n)] \leq D$}}
	\psfrag{w}[][][0.9]{\small{$W$, rate $R$}}
    \psfrag{px}[][][0.9]{\small{$P_{X,Y,Z|Y_p,Y_s}$}}
	\psfrag{enc}[][][0.9]{\small{Encoder}}
	\psfrag{dec}[][][0.9]{\small{Decoder}}
	\psfrag{eve}[][][0.9]{\small{Eaves.}}
	\psfrag{L}[][][0.9]{\small{$\frac{1}{n} I(Y_s^{n};W,Z^{n}) \leq  L$}}
	\includegraphics[width=0.45\textwidth]{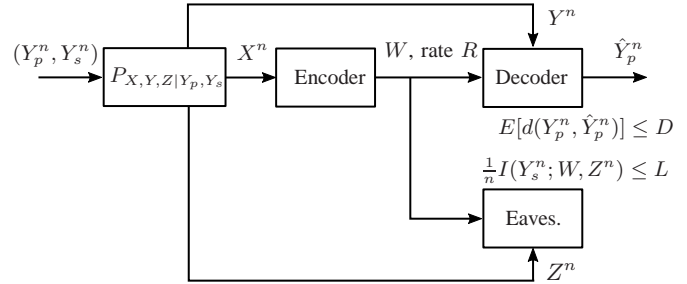}
	\caption{Privacy-constrained remote lossy source coding.}\label{fig:model}
\end{figure}
Our problem is
closely related to works on lossy source coding with a privacy constraint. For instance, Yamamoto considered secure source coding with remote sources to be extracted and protected \cite{yASCP83,yCTFS94}. 
Villard and Piantanida \cite{vpSMSC13} considered  lossy secure source coding  where the source sequence is reconstructed at the decoder satisfying a distortion criterion while limiting the leakage of the source to an eavesdropper below a certain level. Tandon et al. \cite{tspDLSC13} considered privacy of side information at the first decoder against a secondary decoder in the Heegard-Berger setting \cite{hbRDSA85}.  
\cite{srpUPTI13} and \cite{msITPI15} considered one-way and interactive data sharing settings where the source consists of public and private parts and characterized the set of all achievable distortion-leakage pairs. 
The main differences of our work from \cite{srpUPTI13} are the presence of an eavesdropper with correlated side information and the fact that public and private information are considered as remote sources.

\textit{Notation}: We denote the discrete random variables, their
corresponding realizations or deterministic values, and their
alphabets by the upper case, lower case, and calligraphic letters,
respectively. The term $X_{m}^{n}$ denotes the sequence
$\{X_{m},\ldots,X_{n}\}$ when $m\leq n$, and the empty set
otherwise. Also, we use the shorthand notation $X^{n}$ for
$X_{1}^{n}$. The term $X^{n\setminus i}$ denotes the set
$\{X_{1},\ldots,X_{i-1},X_{i+1},\ldots,X_{n}\}$. Cardinality of the
set $\mathcal{X}$ is denoted by $|\mathcal{X}|$. Notation $[1:K]$ denotes the set $\{1,2,\dots,K\}$.  Finally, we use
$X-Y-Z$ to denote that  $(X,Y,Z)$ forms a Markov chain. 
Other notations follow the standard ones in \cite{ekNIT11}.

\section{Secure Remote Source Coding}

\subsection{Problem Formulation}\label{sec:problem_setting}
Let us consider a secure remote source coding shown in Fig.~\ref{fig:model}. Source and side information alphabets, $\mathcal{X}, \mathcal{Y}_p,\mathcal{Y}_s, \mathcal{Y},\mathcal{Z}$ are finite sets. Let $(X^n,Y_p^n,Y_s^n,Y^n,Z^n)$ be $n$-length sequences which have i.i.d. components distributed according to some fixed distribution $P_{X,Y_p,Y_s,Y,Z}$.

The sequence $X^n$ represents the data to be revealed or shared. Public and secret attributes/features associated with the data (but not accessible/allowed to be processed directly) are represented by $Y_p^n$ and $Y_s^n$, respectively. The  rate-limited description $W$  is generated based on $X^n$. 
 The decoder reconstructs the public attribute of the data based on $W$ and correlated side information $Y^n$. For generality, we consider an eavesdropper which has access to the (publicly) stored description and another correlated side information $Z^n$.  The secure remote source coding should ensure the reconstruction quality of the public attribute within a prescribed distortion, and at the same time preserve privacy of the secret part by limiting the amount of information leakage rate at the eavesdropper $\frac{1}{n}I(Y_s^n;W,Z^n)$. We note that if $Z=Y$, the problem  reduces to the case where we impose a privacy constraint against the legitimate receiver.

Let $d: \mathbb{R}\times\mathbb{R} \rightarrow [0,\infty)$ be a distortion measure. The distortion between $Y_p^n$ and its reconstruction $\hat{Y}_p^n$ is defined as
\[d^{(n)}(Y_p^n,\hat{Y}_p^n)= \frac{1}{n}\sum_{i=1}^nd(Y_{p,i},\hat{Y}_{p,i}).\]

We are interested in characterizing the optimal tradeoff of the compression rate, incurred average distortion at the legitimate decoder, and information leakage rate at the eavesdropper.

\begin{definition}\label{def:code_mFAP}
A $(|\mathcal{W}^{(n)}|,n)$-code for secure remote  source coding consists of
\begin{itemize}
  \item an encoder $f^{(n)}: \mathcal{X}^{n} \rightarrow \mathcal{W}^{(n)}$,
  \item a decoder $g^{(n)}: \mathcal{W}^{(n)} \times \mathcal{Y}^{n} \rightarrow \mathcal{\hat{Y}}_p^{(n)}$. 
  \hfill $\lozenge$
\end{itemize}
\end{definition}

\begin{definition}  A rate-distortion-leakage tuple $(R,D,L) \in \mathbb{R}^{3}_{+}$ is said to be \emph{achievable} if, for any $\delta>0$ there exists a sequence of $(|\mathcal{W}^{(n)}|,n)$-codes such that, for all sufficiently large $n$,  
\begin{align}
\frac{1}{n}\log\big|\mathcal{W}^{(n)}\big| &\leq R+\delta,\label{eq:rate_constraint}\\
E[d^{(n)}(Y_p^n,g^{(n)}(W,Y^n))] &\leq D+\delta,\\
\frac{1}{n}I(Y_s^n;W,Z^n) &\leq L+\delta,\label{eq:leakage_constraint}
\end{align}
where $W = f^{(n)}(X^n)$. The \emph{rate-distortion-leakage} region $\mathcal{R}$ is defined as the closure of the set of all achievable tuples.\hfill $\lozenge$
\end{definition}

\subsection{Results}
\begin{theorem}[Inner bound]\label{theorem:inner}
	An inner bound to the rate-distortion-leakage region $\mathcal{R}_{in}$ is given as a set of all tuples $(R,D,L) \in \mathbb{R}_{+}^3$ satisfying 
	\begin{align}
	R &\geq I(X;V|Y) \label{eq:rate}\\
	D&\geq E[d(Y_p,g(V,Y))]\label{eq:distortion}\\
	L &\geq I(Y_s;V,Y) + I(Z;X,Y_s|U)-I(Y;X,Y_s|U) \nonumber\\
	&\qquad -I(X;Z|V,Y_s,Y)+I(X;Y|Y_s,Z),\label{eq:leakage}
	\end{align}
	for some $P_{X,Y_p,Y_s,Y,Z}$$P_{V|X}P_{U|V}$ and $g: \mathcal{V} \times \mathcal{Y} \rightarrow \hat{\mathcal{Y}}_p$ with $|\mathcal{U}|\leq |\mathcal{X}|+3$ and $|\mathcal{V}|\leq (|\mathcal{X}|+3)(|\mathcal{X}|+2)$.
\end{theorem}
\begin{proof}
	The proof is given in Appendix \ref{app:proof_inner} and is based on a  random coding argument where the achievable scheme follows the layered (superposition) coding with binning in \cite{vpSMSC13}. The main difference lies in the analysis of  achievable leakage rate.  
\end{proof}
\begin{remark}
	The constraint in \eqref{eq:leakage} can be rewritten as 
	\begin{align*}
	L &\geq I(X,Y_s;U,V,Y) + I(Z;X,Y_s|U)-I(Y;X,Y_s|U) \\
	&\qquad -I(X;Z|Y_s)-I(X;V|Y,Y_s,Z),
	\end{align*}
		in which the terms on the right-hand side may be interpreted as follows. The term $I(X,Y_s;U,V,Y)$ corresponds to the leakage of $(X,Y_s)$ through the description which depends on the remaining uncertainty at the decoder (which in this case can decode $(U,V)$ and knows $Y$). The terms $I(Z;X,Y_s|U)-I(Y;X,Y_s|U) $ is the additional leakage of $(X,Y_s)$  through the difference of side information available at the eavesdropper and decoder given that the codeword $U$ can be decoded at the eavesdropper. Since we are only interested in the leakage of $Y_s$, the remaining terms correspond to the leakage reduction of $X$ that is ``orthogonal" to that of $Y_s$. 
		The layered coding here provides some degree of freedom to optimize  achievable leakage rate for our general setting. 
		\hfill $\lozenge$
\end{remark}

Next, we provide an outer bound to the rate-distortion-leakage region. 
\begin{theorem}[Outer bound]\label{theorem:outer}
	An outer bound to the rate-distortion-leakage region $\mathcal{R}_{out}$ is given as a set of all tuples $(R,D,L) \in \mathbb{R}_{+}^3$ satisfying \eqref{eq:rate}, \eqref{eq:distortion}, and
	\begin{align}
	L &\geq I(Y_s;V,Y) + I(Z;X,Y_s|U)-I(Y;X,Y_s|U) \nonumber\\
	&\qquad -I(X;Z|V,Y_s,Y)+I(X;Y|T,Y_s,Z),
	\end{align}
	for some $P_{X,Y_p,Y_s,Y,Z}$$P_{T,V|X}P_{U|V}$ and $g: \mathcal{V} \times \mathcal{Y} \rightarrow \hat{\mathcal{Y}}_p$. 
\end{theorem}
\begin{proof}
	The proof is based on standard properties of the entropy function and the Csiszar's sum identity \cite{ekNIT11} and is given in Appendix \ref{app:proof_outer}. 
\end{proof}

\begin{remark}
	The results in Theorems \ref{theorem:inner} and \ref{theorem:outer} can be extended to a scenario where the sequences  $(X^n,Y_p^n,Y_s^n)$ are available directly at the encoder. In this case, we can replace $X$ by $(X,Y_p,Y_s)$ in Theorems \ref{theorem:inner} and \ref{theorem:outer}, and the joint distributions become those of the form $P_{X,Y_p,Y_s,Y,Z}P_{V|X,Y_p,Y_s}P_{U|V}$ and $P_{X,Y_p,Y_s,Y,Z}P_{T,V|X,Y_p,Y_s}P_{U|V}$, respectively.

			Theorems \ref{theorem:inner} and \ref{theorem:outer} can also be generalized to the case where $Y_p$ and $Y_s$ are not  ``disjoint," i.e., they share some common part. For instance, the decoder may wish to reconstruct some attributes associated with the data that are considered as secret to the eavesdropper. We simply modify the setup by replacing $Y_p$ by $(Y_p,Y_c)$ and $Y_s$ by $(Y_s,Y_c)$, where $Y_c$ acts as a common part that is supposed to be reconstructed at the decoder and protected against the eavesdropper.  Theorems~\ref{theorem:inner} and \ref{theorem:outer} continue to hold with $Y_p$ replaced by $(Y_p,Y_c)$ and $Y_s$ replaced by $(Y_s,Y_c)$ and the joint distribution of relevant source and side information is given by $P_{X,Y_p,Y_s,Y_c,Y,Z}$.\hfill $\lozenge$
		\end{remark}
				
We see that inner and outer bounds in Theorems \ref{theorem:inner} and \ref{theorem:outer} do not match in general. In particular, there is a gap between the leakage rate bounds. 
The difficulty of proving the tight bound lies in the complex dependency of information available at the eavesdropper and the secret remote source $Y_s^n$.
Nevertheless, there exist some special cases where the bounds are tight.

\begin{corollary}\label{corollary:region1}
For the sources and side information whose joint distributions satisfy $I(X;Y_s,Z|Y)=0$, the \emph{rate-distortion-leakage} region $\mathcal{R}$ is given as a set of all tuples $(R,D,L) \in \mathbb{R}_{+}^3$ satisfying \eqref{eq:rate}, \eqref{eq:distortion}, and $L \geq I(Y_s;Z)$ 
	for some $P_{X,Y_p,Y_s,Y,Z}$$P_{V|X}$ and $g: \mathcal{V} \times \mathcal{Y} \rightarrow \hat{\mathcal{Y}}_p$ with  $|\mathcal{V}|\leq |\mathcal{X}|+1$.
	
		Interestingly, in this case, the only leakage of $Y_s^n$  is from correlated side information $Z^n$. There is no additional leakage rate from the source description $W$ sent over the rate-limited link. 
			This can be explained as follows. Our achievability scheme is based on the binning technique which renders $I(W;Y^n) \leq n\delta_{\epsilon}$. Since we have the Markov chain $W-X^n-Y^n-(Y_s^n,Z^n)$, from the data processing inequality, the additional leakage rate of $Y_s^n$ due to $W$ becomes negligible, i.e., $I(Y_s^n;W|Z^n) \leq I(Y_s^n,Z^n;W)\leq I(Y^n;W)\leq n\delta_{\epsilon}$.
\end{corollary}

\begin{proof}
	The proof follows from specializing Theorems \ref{theorem:inner} and \ref{theorem:outer} to the case where $X-Y-(Y_s,Z)$ forms a Markov chain, where in achievability, we  choose $U =\emptyset$. 
\end{proof}

\begin{corollary}\label{corollary:region}
	For the sources and side information whose joint distributions satisfy $I(X;Y|Y_s,Z)=0$, the \emph{rate-distortion-leakage} region $\mathcal{R}$ is given as a set of all tuples $(R,D,L) \in \mathbb{R}_{+}^3$ satisfying \eqref{eq:rate}, \eqref{eq:distortion}, and
	\begin{align}
		L &\geq I(Y_s;V,Y) + I(Z;X,Y_s|U)-I(Y;X,Y_s|U) \nonumber\\
		&\qquad -I(X;Z|V,Y_s,Y),
	\end{align}
		for some $P_{X,Y_p,Y_s,Y,Z}$$P_{V|X}P_{U|V}$ and $g: \mathcal{V} \times \mathcal{Y} \rightarrow \hat{\mathcal{Y}}_p$ with $|\mathcal{U}|\leq |\mathcal{X}|+3$ and $|\mathcal{V}|\leq (|\mathcal{X}|+3)(|\mathcal{X}|+2)$.
\end{corollary}
\begin{proof}
	The achievability proof follows directly from Theorem \ref{theorem:inner} with $I(X;Y|Y_s,Z)=0$, while the converse follows from Theorem \ref{theorem:outer} and the fact that $I(X;Y|T,Y_s,Z) \geq 0$.
\end{proof}
\begin{remark}
	Corollaries \ref{corollary:region1} and \ref{corollary:region} hold also for the  case of stochastic encoder where the description $W$ is randomly generated according to a conditional PMF $p(w|x^n)$. This follows from the fact that in the proof of Theorem~\ref{theorem:outer}, we do not make any assumption regarding the deterministic encoder.
	\hfill $\lozenge$
\end{remark}
\begin{remark}
	We see that the rate-distortion-leakage region is known for several classes of sources and side information, e.g., those satisfying  $I(X;Y|Y_s,Z)=0$ in Corollary~\ref{corollary:region} which includes also the semi-deterministic mapping with $X$ being a deterministic function of $Y_s$. 
	Moreover, the result in Corollary~\ref{corollary:region}  recovers several existing results in the secure lossy source coding literature, e.g., 
	\begin{itemize}
		\item when $Y=Z$, we may think of the privacy leakage constraint as one imposed at the legitimate decoder.  If $X=(Y_p,Y_s)$, Corollary \ref{corollary:region} recovers the result of the utility-privacy tradeoff  with side information in \cite{srpUPTI13}. Furthermore, if $X=(Y_p,Y_s)$ and $Y=Z=\emptyset$, 
		Corollary~\ref{corollary:region} recovers the result in \cite{yASCP83}. 
		\item when $X=Y_p=Y_s$, Corollary \ref{corollary:region} recovers the result of secure lossy source coding problem studied in \cite{vpSMSC13}.
		\item when $Y=Y_s$, the leakage term becomes $I(Y_s^n;W,Z^n)$ which is of the same type as the side information privacy considered in \cite{tspDLSC13}. The main difference is that  in \cite{tspDLSC13}  side information privacy is considered at the secondary receiver who observes no additional side information. If the reconstruction constraint at the secondary receiver is neglected, then  zero leakage rate is achievable by the Wyner-Ziv coding \cite{wzTRDF76}. 
		Interestingly, in our case where the eavesdropper has access to the additional side information, the layered random binning scheme turns out to be optimal. An achievable leakage rate in this case is $I(Y_s;U,Z)+I(V;Z|Y_s,U)$ which is larger than $I(Y_s;Z)$. This is due to the fact that conditioned on $Z^n$, the Wyner-Ziv bin indices are still correlated with side information $Y_s^n$, and thus revealing some information about $Y_s^n$ to the eavesdropper.  \hfill $\lozenge$
	\end{itemize}
\end{remark}

\begin{corollary}\label{corollary:lossless}
	When we set $Y_p=X$ and consider a lossless reconstruction of $X^n$ at the decoder, an inner bound to the rate-leakage region is given by the set of all $(R,L)\in \mathbb{R}_{+}^2$ satisfying
	\begin{align*}
	R &\geq H(X|Y)\\
	L &\geq I(Y_s;X,Y) + I(Z;X,Y_s|U) - I(Y;X,Y_s|U)\\&\qquad +I(X;Y|Y_s,Z),
	\end{align*}
	for some $P_{X,Y_s,Y,Z}P_{U|X}$ with $|\mathcal{U}|\leq |\mathcal{X}|$.
	
The inner bound above can be proved similarly as in Theorem \ref{theorem:inner}. In fact, it can be obtained  from Theorem~\ref{theorem:inner} by setting $Y_p=X=V$.
\end{corollary}

\begin{remark}
 We note that the special case of lossless reconstruction above was considered recently in \cite{aalLSSC15} where an inner bound to the rate-equivocation region is provided. In general, the results in  Corollary \ref{corollary:lossless} and \cite[Theorem 3]{aalLSSC15} do not match. As an example where $Y=Y_s$ and $Z=\emptyset$, it can be shown that Corollary \ref{corollary:lossless} implies that zero leakage rate is achievable (by choosing $U=\emptyset$). However, the achievable leakage rate according to \cite[Theorem 3]{aalLSSC15} can be strictly positive. \hfill $\lozenge$
\end{remark}

\subsection{Quadratic Gaussian Example}
We consider an example of the tradeoff in Corollary \ref{corollary:region} for Gaussian sources under quadratic distortion. 
Assuming that $X \sim \mathcal{N}(0,N_x)$, $Y_s= X + \tilde{N}_s$, where $\tilde{N}_s \sim \mathcal{N}(0,N_s) \perp X$, and $Y_p= Y_s + \tilde{N}_p$, where $\tilde{N}_p \sim \mathcal{N}(0,N_p-N_s) \perp Y_s$. Note that $X-Y_s-Y_p$ forms a Markov chain. Also, we assume that there is no side information, i.e., $Y=Z=\emptyset$.  The tradeoff region in Corollary \ref{corollary:region} reduces to the set of all $(R,D,L)$ satisfying
\begin{align}
R &\geq \frac{1}{2}\log\Big(\frac{N_x}{D-N_p}\Big) \label{eq:R_Guassian},\\
L &\geq \frac{1}{2}\log\Big(\frac{N_x+N_s}{D-N_p+N_s}\Big),\label{eq:L_Gaussian}
\end{align}
for $D > N_p$. 

While our main results were proven for discrete memoryless sources, the extension to the quadratic Gaussian case is standard
and it follows, e.g., \cite{wTRDF78} and 
\cite{ekNIT11}. For achievability, we set $U=\emptyset$, choose $V$ to be jointly Gaussian with $X$, i.e., $V=X+Q, Q \sim \mathcal{N}(0,N_q)$, and choose the reconstruction function $g(\cdot)$ to be an MMSE estimate of $Y_p$ given $V$. By letting $D=E[d(Y_p,g(V))]$ and substituting it into the constraints on $R$ and $L$, we obtain the result above. The converse follows from utilizing the EPI \cite{ekNIT11} together with the fact that $R$ and $L$ are decreasing in $h(X^n|W)$ and $h(Y_s^n|W)$, respectively. For the more detailed proof, please see Appendix~\ref{app:proof_Gaussian_example}. 

From \eqref{eq:R_Guassian} and \eqref{eq:L_Gaussian}, we can also write the minimum distortion as a function of $R$ and $L$, i.e., \[D_{\text{min}}(R,L)=\max\{N_p+N_x2^{-2R},N_p-N_s+(N_x+N_s)2^{-2L}\}.\]
For a fixed $R$, the minimum distortion  $D_{\text{min}}$ decreases with $L$, illustrating the utility-privacy tradeoff in terms of minimum achievable distortion and information leakage rate.
\section{Logarithmic loss distortion}
 Logarithmic loss  \cite{mfUP98,cwMSCU14} is a measure of quality of the ``soft" estimate used in several applications 
 \cite{bPRML06},\cite{aabgEBDM07}.
 Under logarithmic loss, we employ a soft estimate of the source in terms of a probability distribution over the source alphabet.
For a sequence $\hat{X}^n \in \hat{\mathcal{X}}^n$, we denote $\hat{X}_i$, $i=1,\ldots,n$, the $i^{th}$ element of $\hat{X}^n$. Then $\hat{X}_i, i=1,\ldots,n$ is a probability distribution on $\mathcal{X}$, i.e., $\hat{X}_i: \mathcal{X} \rightarrow [0,1]$, and $\hat{X}_i(x)$ is a probability distribution on $\mathcal{X}$ evaluated for the outcome $x \in \mathcal{X}$. 

\begin{definition}[Logarithmic loss \cite{cwMSCU14}]
	The logarithmic loss distortion is defined as $d(x,\hat{x}) = \log(\frac{1}{\hat{x}(x)})= D_{KL}(\mathbf{1}_{\{x\}}||\hat{x})$, where $\mathbf{1}_{\{x\}}: \mathcal{X} \rightarrow \{0,1\}$ is an indicator function such that, for $a \in \mathcal{X}$, $\mathbf{1}_{\{x\}}(a)=1$ if  $a=x$, and $\mathbf{1}_{\{x\}}(a)=0$ otherwise. 
Using this definition for symbol-wise distortion, it is standard to define the distortion between sequences as $d^{(n)}(x^n,\hat{x}^n)=\frac{1}{n}\sum_{i=1}^n d(x_i,\hat{x}_i)$.
\end{definition}

Under logarithmic loss, the average distortion and conditional entropy (equivocation at decoder) are closely related (see, e.g., Lemma 
\ref{lemma:loglossouterbound} below \cite{cwMSCU14}). 
This is reminiscent of what is known, e.g., for the Gaussian settings under quadratic distortion. 

\begin{lemma}\label{lemma:loglossouterbound}
	Let $C=(W,Y^n)$ be the argument of the reconstruction function $g^{(n)}(\cdot)$, i.e., $\hat{X}^n = g^{(n)}(C)$, then under the logarithmic loss distortion measure, we get $E[d^{(n)}(X^n,g^{(n)}(C))] \geq \frac{1}{n}\sum_{i=1}^nH(X_i|C)$.
\end{lemma}
The proof of the lemma  follows from definition of logarithmic loss, i.e., $d(x_i,g^{(n)}_i(c)) \triangleq \log(\frac{1}{q(x_i|c)})$ where $q$ is a probability measure on $\mathcal{X}$. Then  the expected distortion conditioned on $C=c$,
\begin{align*}
&E[d^{(n)}(X^n,g^{(n)}(c))|C=c]\\
& = E[\frac{1}{n}\sum_{i=1}^n d(X_i,g^{(n)}_i(c))|C=c]\\
&= \sum_{x^n \in \mathcal{X}^n} p(x^n|c) \frac{1}{n}\sum_{i=1}^n d(x_i,g^{(n)}_i(c))\\
&= \frac{1}{n}\sum_{i=1}^n \sum_{x_i \in \mathcal{X}} p(x_i|c) \log\bigg(\frac{1}{q(x_i|c)}\bigg)\\
&= \frac{1}{n}\sum_{i=1}^n \sum_{x_i \in \mathcal{X}} p(x_i|c) \log\bigg(\frac{p(x_i|c)}{q(x_i|c)}\cdot \frac{1}{p(x_i|c)}\bigg)\\
&= \frac{1}{n} \sum_{i=1}^n D_{KL}(p(x_i|c)||q(x_i|c)) + \frac{1}{n} \sum_{i=1}^nH(X_i|C=c)\\
& \geq  \frac{1}{n} \sum_{i=1}^n H(X_i|C=c).
\end{align*}
By averaging both sides over all $c \in \mathcal{C}$, from the law of total expectation, the lemma is proved.

We note that Lemma~\ref{lemma:loglossouterbound} holds only for the case of symbol-by-symbol logloss distortion, i.e., $d^{(n)}(x^n,\hat{x}^n)=\frac{1}{n}\sum_{i=1}^n d(x_i,\hat{x}_i)$, which is considered in the problem formulation. 
\begin{corollary}\label{corollary:region_logloss}
	Consider the privacy-constrained remote source coding in Fig. \ref{fig:model}. When the decoder and eavesdropper observe identical side information, i.e., $Y= Z$ (can be seen as privacy leakage against the legitimate receiver), the \emph{rate-distortion-leakage} region under logarithmic loss is given as a set of all tuples $(R,D,L) \in \mathbb{R}_{+}^3$ satisfying 
	\begin{align}
	R &\geq I(X;V|Y) \\
	D&\geq H(Y_p| V,Y)\\
	L &\geq I(Y_s;V,Y),
	\end{align}
	for some $P_{X,Y,Y_p,Y_s}$$P_{V|X}$, with $|\mathcal{V}|\leq |\mathcal{X}|+2$.
\end{corollary}

\begin{proof}
For achievability, we apply Corollary \ref{corollary:region} under logarithmic loss distortion with $Y=Z$. We choose $g(\cdot)$ to be a conditional probability distribution on $\mathcal{Y}_p$, i.e., $g(v,y)=p(y_p|v,y)$ which gives $E[d(Y_p,g(V,Y))]=H(Y_p|V,Y)$.\footnote{The proof of  Theorem~\ref{theorem:inner}  holds for bounded distortion measures. However, 
	it can still be extended to logarithmic loss distortion by perturbing the reconstruction probability distribution (see e.g., \cite[Remark 3.4]{ckOSSC13}).}	
The converse follows by setting $V_i = (W,Y^{n \setminus i})$ and applying  Lemma~\ref{lemma:loglossouterbound} which gives $E[d^{(n)}(Y_{p}^n,g^{(n)}(W,Y^n))]=\frac{1}{n}\sum_{i=1}^n E[d(Y_{p,i},g_i^{(n)}(W,Y^n))]\geq \frac{1}{n}\sum_{i=1}^nH(Y_{p,i}|W,Y^n)$.
\end{proof}

\subsection{Secure Information Bottleneck}
Based on Corollary \ref{corollary:region_logloss}, we can formulate an optimization problem where for a given $P_{X,Y,Y_p,Y_s}$, we wish to minimize the rate $I(X;V|Y)$ over all $P_{V|X}$ subject to the constraints on the distortion and information leakage, i.e., 
\begin{align}
\min_{P_{V|X}}  I(X;V|Y)&\\
\text{s.t.} \ I(Y_p;V,Y) &\geq D'\label{eq:info_bottleneck_information}\\
I(Y_s;V,Y)&\leq L,\label{eq:info_bottleneck_leakage}
\end{align}
where in \eqref{eq:info_bottleneck_information}, the distortion constraint is rewritten as the ``information" constraint with $D' \triangleq H(Y_p)-D$. 
The minimum above corresponds to the minimum achievable rate for fixed $D'$ and $L$ and may be termed as the rate-information-leakage function $R_{\min}(D',L)$.  
We can see that the optimization problem is not convex, e.g., $I(Y_p;V,Y)$ is a convex function in $P_{V|X}$ for a fixed $P_{X,Y,Y_p,Y_s}$. Interestingly, the problem shares some similarity with the information bottleneck problem \cite{tpbTIBM00} where in our case there is an additional constraint on the privacy leakage rate \eqref{eq:info_bottleneck_leakage} and the presence of side information $Y$. Given that some side information $Y$ is known beforehand, the goal here is to represent the data $X$ efficiently by a compact representation $V$, while maximizing the relevance of $V$ on the public attribute $Y_p$ and minimizing the relevance of $V$ on the secret attribute $Y_s$. Due to the additional constraint on the privacy leakage rate, we term this optimization problem as  \emph{secure information bottleneck}. When $Y=\emptyset$, it also corresponds to a variant of information bottleneck  considered in \cite{ctERSW02}. The secure remote source coding problem
under logarithmic loss distortion therefore gives an operational meaning to the secure information bottleneck.

To solve the secure information bottleneck problem, we may extend the iterative algorithm proposed for the information bottleneck \cite{tpbTIBM00}. 
For example, the optimization problem above can be solved by minimizing the function
\begin{align*}
\mathcal{L} &= I(X;V|Y)-\beta[I(Y_p;V,Y)-\gamma I(Y_s;V,Y)]\\
&\qquad -\sum_{x,v}\lambda(x)p(v|x)
\end{align*}
over all $p(v|x)$, where $\lambda(x)$ are functions of $x$, and parameters $\beta$ and $\gamma$ are introduced to capture the tradeoff due to information (distortion) and privacy leakage constraints.
Similarly as in \cite{tpbTIBM00}, \cite{ctERSW02}, it can be shown that the stationary points of $\mathcal{L}$ are given by
 \begin{align}
 p(v|x) &= \frac{1}{Z(x,\beta,\gamma)}\exp\Big\{\sum_y p(y|x)\Big[\log p(v|y)\nonumber\\ &\qquad -\beta\big[D_{KL}(p(y_p|x,y)||p(y_p|v,y))\nonumber\\
 &\qquad -\gamma D_{KL}(p(y_s|x,y)||p(y_s|v,y)) \big]\Big]\Big\}, \label{eq:stationary_point}
 \end{align}
where $Z(x,\beta,\gamma)$ is the normalization term satisfying $\sum_v p(v|x)=1$. 

Equation \eqref{eq:stationary_point} together with the marginalization constraints for $p(v|y)$, $p(y_p|v,y)$, and $p(y_s|v,y)$ forms a set of equations which, given initial distributions, can be solved iteratively similarly as in \cite{tpbTIBM00},\cite{ctERSW02}. With $\beta$ and $\gamma$ that admit feasible solutions, the algorithm converges to a stationary point which may not be the global optimum. 

Alternatively, the problem can also be solved by a heuristic method, e.g., by extending the agglomerative information bottleneck \cite{stAIB99},\cite{msfmFIBT14} to include the privacy leakage constraint. For example, Algorithm $1$ in \cite{msfmFIBT14} can be modified to include a condition that requires the merging indices $i$ and $j$ to satisfy both information and privacy leakage constraints.  
\subsection{Example}
The secure information bottleneck may alternatively be formulated as 
\begin{align}
\max_{P_{V|X}}  I(Y_p;V,Y)& \label{eq:info_bottleneck_dis}\\
\text{s.t.} \ I(X;V|Y)&\leq R\\
I(Y_s;V,Y)&\leq L
\end{align}
for which the maximum in \eqref{eq:info_bottleneck_dis} corresponds to the maximum achievable information for given  $R$ and $L$ and may be termed as the  information-rate-leakage function $D'_{\max}(R,L)$. 

In the following, we consider two simple examples under the assumptions that $Y=\emptyset$  and $X-Y_s-Y_p$ forms a Markov chain, in which we can express $D'_{\max}(R,L)$ in closed form.

\emph{(i) Binary source}:
  Let $X \sim \text{Bernoulli}(1/2)$, $Y_s$ be an output of a BSC($p$), $p \in [0,1/2]$, with input $X$, and $Y_p$ be an output of a BSC($q$), $q \in [0,1/2]$, with input $Y_s$. The maximum achievable information for  given  $R \in [0, H(X)]$ and  $L \in [0, H(Y_s)]$ is given by
\begin{align}
D'_{\max}(R,L) &= H(Y_p)-\max\Big\{h(h^{-1}(H(Y_s)-L)*q), \nonumber \\ &\qquad \qquad h(h^{-1}(H(X)-R)*p*q)\Big\},
\end{align}
where $h(\cdot)$ is a binary entropy function with the inverse $h^{-1} : [0, 1] \rightarrow [0, 1/2]$, and $a*b \triangleq a(1-b)+(1-a)b$. 

The achievability proof follows by letting $V$ be an output of a BSC with input $X$. Then letting $D'=I(Y_p;V)=H(Y_p)-h(h^{-1}(H(X|V))*p*q) = H(Y_p)-h(h^{-1}(H(Y_s|V))*q)$. We obtain the result above by substituting $D'$ in the constraints on $R$ and $L$ and using the fact that $h(h^{-1}(u)*q)$ is an increasing function in $u$ for $q \in [0,1/2]$. The converse follows from Mrs. Gerber's lemma \cite{ekNIT11}. For the more detailed proof, please see Appendix \ref{app:proof_secureIB_binary_example}.

\emph{(ii) Gaussian source}: Let $(X,Y_s,Y_p)$ be jointly Gaussian, i.e., $X \sim \mathcal{N}(0,N_x)$, $Y_s= X + \tilde{N}_s$, where $\tilde{N}_s \sim \mathcal{N}(0,N_s) \perp X$, and $Y_p= Y_s + \tilde{N}_p$, where $\tilde{N}_p \sim \mathcal{N}(0,N_p-N_s) \perp Y_s$. The minimum achievable distortion for given  $R$ and  $L$ is given by
\begin{align}\label{eq:distortion-rate-leakage_function-Gaussian}
D'_{\max}(R,L) &= \min\Big\{\frac{1}{2}\log\Big(\frac{N_x+N_p}{N_x2^{-2R}+N_p}\Big), \nonumber \\ &\qquad \frac{1}{2}\log\Big(\frac{N_x+N_p}{(N_x+N_s)2^{-2L}+N_p-N_s}\Big) \Big\}.
\end{align}
The proof follows similarly as in the binary case where in achievability we let $V$ be jointly Gaussian with $X$, i.e., $V=X+Q$, $Q\sim \mathcal{N}(0,N_q)$, and in the converse we use the conditional EPI \cite{ekNIT11}. Note that if the leakage constraint is neglected, e.g., letting $L \rightarrow \infty$, \eqref{eq:distortion-rate-leakage_function-Gaussian} reduces to the optimal information-rate function in \cite{wmRIOG14}.


\appendices
\section{Proof of Theroem \ref{theorem:inner}}\label{app:proof_inner}
The proof is based on the achievable scheme used in \cite{vpSMSC13} which is a layered (superposition) coding scheme with binning. 

\begin{figure*}[!t]
	\normalsize
	\setcounter{MYtempeqncnt}{\value{equation}}
	\setcounter{equation}{5}
	\begin{align*}
	&	H(X^n|J,K,Y_s^n,Z^n) \\
	&\overset{(a)}{\leq} H(X^n,E|J,K,U^n(J),V^n(J,K),Y_s^n,Z^n)\\
	&\leq H(X^n|U^n,V^n,Y_s^n,Z^n,E) + H(E)\\
	&= \mathrm{Pr}(E=0) H(X^n|U^n,V^n,Y_s^n,Z^n,E=0) + \mathrm{Pr}(E=1) H(X^n|U^n,V^n,Y_s^n,Z^n,E=1) + H(E)\\
	&\overset{(b)}{\leq}H(X^n|U^n,V^n,Y_s^n,Z^n,E=0) +\delta_{\epsilon} H(X^n)+ h(\delta_{\epsilon}) \\
	&\leq H(X^n|U^n,V^n,Y_s^n,Z^n,E=0) + n\delta_{\epsilon} \log|\mathcal{X}| + h(\delta_{\epsilon})\\
	&= \sum_{(u^n,v^n,y_s^n,z^n) \in \mathcal{T}_{\epsilon}^{(n)}} p(u^n,v^n,y_s^n,z^n|E=0) H(X^n|U^n=u^n,V^n=v^n,Y_s^n=y_s^n,Z^n=z^n,E=0)  + n\delta_{\epsilon} \log|\mathcal{X}| + h(\delta_{\epsilon})\\
	&\overset{(c)}{\leq} \sum_{(u^n,v^n,y_s^n,z^n) \in \mathcal{T}_{\epsilon}^{(n)}} p(u^n,v^n,y_s^n,z^n|E=0) \log|\mathcal{T}_{\epsilon}^{(n)}(X|u^n,v^n,y_s^n,z^n)| + n\delta_{\epsilon} \log|\mathcal{X}|  + h(\delta_{\epsilon})\\&\overset{(d)}{\leq} n(H(X|U,V,Y_s,Z)+\delta_{\epsilon}')
	\end{align*}
	\setcounter{equation}{\value{MYtempeqncnt}}
	\hrulefill
	\vspace*{4pt}
\end{figure*}
For random codebook generation, we fix $P_{V|X}P_{U|V}$ and the reconstruction function $g(\cdot)$. 
\begin{itemize}
	\item Randomly and independently generate codewords $u^n(j)$ for $j \in [1:2^{n(I(X;U)+\delta_{\epsilon})}]$ according to the product distribution $\prod_{i=1}^n P_U(u_i)$. We distribute the codewords uniformly at random into  $2^{n(I(X;U|Y)+2\delta_{\epsilon})}$ equal-sized bins $b_U(w_1)$,  $w_1 \in [1:2^{n(I(X;U|Y)+2\delta_{\epsilon})}]$. Each bin contains $2^{n(I(U;Y)-\delta_{\epsilon})}$ codewords, each indexed by $w'$. There exists a one-to-one mapping between
index $j$ and the pair of bin/codeword indices $(w_1,w')$ such that, without loss of generality, we can
identify $j = (w_1,w')$.
	\item For each $j$,  randomly and conditionally independently generate codewords $v^n(j,k)$ where $k\in [1:2^{n(I(X;V|U)+\delta_{\epsilon})}]$ according to the conditional product distribution $\prod_{i=1}^n P_{V|U}(v_i|u_i(j))$, and distribute these codewords uniformly at random into $2^{n(I(X;V|U,Y)+2\delta_{\epsilon})}$ equal-sized bins $b_V(j,w_2)$, $w_2 \in [1:2^{n(I(X;V|U,Y)+2\delta_{\epsilon})}]$. Each bin $b_V(j,w_2)$ contains $2^{n(I(V;Y|U)-\delta_{\epsilon})}$ codewords, each indexed by $w''$. There exists a one-to-one mapping between	index $k$ and the pair of bin/codeword indices $(w_2,w'')$ such that, without loss of generality, we can
	identify $k = (w_2,w'')$.
	\item The codebooks are revealed to all parties.
\end{itemize}

For encoding, 
\begin{itemize}
	\item Given $X^n=x^n$, the encoder looks for codeword  $u^n(j)$ such that $(x^n,u^n(j))$ are jointly typical. From the covering lemma \cite{ekNIT11}, with high probability, there exist such a codeword since there are more than $2^{nI(X;U)}$ codewords $u^n(j)$. If there are more than one, the encoder selects one with the smallest index $j$.
	\item Then based on $x^n$ and the chosen $u^n(j)$, the encoder looks for a codeword $v^n(j,k)$ such that $(x^n,u^n(j)),v^n(j,k)$ are jointly typical. From the covering lemma, with high probability,  there exists such a codeword since there are more than $2^{nI(X;V|U)}$ codewords $v^n(j,k)$ for each $j$. If there are more than one, the encoder selects one with the smallest index $k$. 
	\item The encoder sends the bin indices $w_1$ and $w_2$ of the chosen codewords to the decoder. The total rate is thus equal to $I(X;U|Y)+I(X;V|U,Y)+4\delta_{\epsilon}=I(X;V|Y)+4\delta_{\epsilon}$, where the equality follows from  the Markov chain $U-V-X-Y$.
\end{itemize} 

For decoding,
\begin{itemize}
	\item Based on $Y^n=y^n$ and the bin indices  $(w_1,w_2)$, the decoder looks for a unique codeword $u^n(j)$ in bin $b_U(w_1)$ such that $(y^n,u^n(j))$ are jointly typical. From the packing lemma \cite{ekNIT11}, there exists such a codeword $u^n(j)$ with high probability since there are less than $2^{nI(U;Y)}$ codewords $u^n(j)$ in each bin $b_U(w_1)$.
	\item Then based on the decoded $u^n(j)$, the decoder looks for a unique codeword $v^n(j,k)$ in bin $b_V(j,w_2)$ such that $(y^n,u^n(j),v^n(j,k))$ are jointly typical. From the packing lemma, there exists such a codeword $v^n(j)$ with high probability since there are less than $2^{nI(V;Y|U)}$ codewords $v^n(j,k)$ in each bin $b_V(j,w_2)$.
	\item The decoder reconstructs $\hat{y}_p^n$ such that $\hat{y}_{p,i}=g(v_i(j,k),y_i)$ for $i=1,\ldots,n$.
\end{itemize}

Let $J=(W_1,W')$ and $K=(W_2,W'')$ be the indices associated with the chosen codewords $U^n(J)$ and $V^n(J,K)$. From LLN, we have that the sequences  $(X^n,Y^n,Y_p^n,Y_s^n,Z^n,U^n(J),V^n(J,K))$ are jointly typical with high probability. Thus, using similar arguments as in \cite{wzTRDF76} for a bounded distortion measure, the distortion constraint is satisfied if $D \geq E[d(Y_p,g(V,Y))]$.

Before proceeding with the analysis of the leakage rate, we give a lemma which provides a bound on the
$n$-letter conditional entropy based on properties of jointly typical sequences.
\begin{lemma} \label{lemma:lower_bound}
	Let the index $J$ and the pair of indices $(J,K)$ be the indices specifying codewords $U^n$ and $V^n$, respectively. If $\mathrm{Pr}((X^n,U^n(J),V^n(J,K),Y_s^n,Z^n) \in \mathcal{T}_{\epsilon}^{(n)}) \rightarrow 1$ as $n \rightarrow \infty$, we have that $H(X^n|J,K,Y_s^n,Z^n) \leq n(H(X|U,V,Y_s,Z)+\delta_{\epsilon})$.
\end{lemma}

	To prove the lemma, let $E$ be a binary random variable taking value $0$ if $(X^n,U^n(J),V^n(J,K),Y_s^n,Z^n) \in \mathcal{T}_{\epsilon}^{(n)}$, and $1$ otherwise. Since $(X^n,U^n(J),V^n(J,K),Y_s^n,Z^n) \in \mathcal{T}_{\epsilon}^{(n)}$ with high probability, we have $\mathrm{Pr}(E=1) \leq \delta_{\epsilon}$. The proof steps are given on top of the previous page
	where step $(a)$ follows from the fact that given the codebook, $U^n$ and $V^n$ are functions of $J$ and $(J,K)$, $(b)$ follows from $\mathrm{Pr}(E=1) \leq \delta_{\epsilon}$ where $h(\cdot)$ is the binary entropy function, and $(c)$ and $(d)$ follow from the property of jointly typical set \cite{ekNIT11} with $\delta_{\epsilon}, \delta_{\epsilon}' \rightarrow 0$ as $\epsilon \rightarrow 0$, and $\epsilon \rightarrow 0$ as $n \rightarrow \infty$.

The privacy leakage average over all randomly chosen codebooks can be bounded as follows. 
\begin{align*} \label{eq:bound_step}
&I(Y_s^n;W_1,W_2,Z^n)= H(Y_s^n) -H(Y_s^n|W_1,W_2,Z^n) \\
&=H(Y_s^n)-H(Y_s^n,X^n|W_1,W_2,Z^n) \\ &\qquad +H(X^n|W_1,W_2,Y_s^n,Z^n) \\
&\leq H(Y_s^n)-H(Y_s^n,X^n|J,Z^n)+H(W_2) \\ &\qquad +H(X^n|J,K,Y_s^n,Z^n) \\ &\qquad +I(X^n;W',W''|W_1,W_2,Y_s^n,Z^n) \\
&\leq H(Y_s^n)-H(Y_s^n,X^n,Z^n)+H(J)+H(Z^n|J)+H(W_2) \\&\qquad +H(X^n|J,K,Y_s^n,Z^n) +H(W',W''|W_1,W_2,Y_s^n,Z^n) \\
&\overset{(a)}{\leq} n[-H(X,Z|Y_s)+I(X;U)+H(Z|U)+I(X;V|U,Y) \\& \qquad +H(X|U,V,Y_s,Z)+\delta_{\epsilon}']  \\&\qquad  +H(W',W''|W_1,W_2,Y_s^n,Z^n)\\
&\overset{(b)}{\leq}  n[P+\delta_{\epsilon}'] + I(W',W'';Y^n|W_1,W_2,Y_s^n,Z^n) + n\epsilon_n  \\
&\leq n[P+\delta_{\epsilon}'] + H(Y^n|Y_s^n,Z^n) - H(Y^n|J,K,Y_s^n,Z^n)+ n\epsilon_n  \\
&\overset{(c)}{\leq} n[P+H(Y|Y_s,Z)+ \delta_{\epsilon}'-H(Y|U,V,Y_s,Z)+\delta_{\epsilon}]+ n\epsilon_n  \\
&\overset{(d)}{=} n[I(Y_s;V,Y) + I(Z;X,Y_s|U)-I(Y;X,Y_s|U)  \\
&\qquad -I(X;Z|V,Y_s,Y)+I(X;Y|Y_s,Z) + \delta_{\epsilon}''] \\
&\leq n[L+ \delta_{\epsilon}'']  
\end{align*}
if $L \geq I(Y_s;V,Y) + I(Z;X,Y_s|U)-I(Y;X,Y_s|U) -I(X;Z|V,Y_s,Y)+I(X;Y|Y_s,Z)$, 
where $(a)$ follows from memoryless property of the sources, from the codebook generation with $J\in[1:2^{n(I(X;U)+\delta_{\epsilon})}]$, $W_2 \in [1:2^{n(I(X;V|U,Y)+2\delta_{\epsilon})}]$, and from the bounds $H(Z^n|J) \leq n(H(Z|U)+\delta_{\epsilon})$ and $H(X^n|J,K,Y_s^n,Z^n) \leq n(H(X|U,V,Y_s,Z)+\delta_{\epsilon})$ which can be shown similarly as in Lemma \ref{lemma:lower_bound}, $(b)$ follows from defining $P = -H(X,Z|Y_s)+I(X;U)+H(Z|U)+I(X;V|U,Y) +H(X|U,V,Y_s,Z)$ and Fano's inequality $H(W',W''|W_1,W_2,Y^n,Y_s^n,Z^n) \leq n\epsilon_n$ which holds since given the codebook and $W_1,W_2,Y^n$, the decoder can decode $(W',W'')$ with high probability, $(c)$ follows from Lemma~\ref{lemma:entropy_bound} below, and $(d)$ follows from the definition of $P$ and the Markov chain $U-V-X-(Y,Y_p,Y_s,Z)$.

\begin{lemma}\label{lemma:entropy_bound}
	Given the codebook where $J$ is the codeword index of $U^n$, and $(J,K)$ is the codeword index of $V^n$, if $\mathrm{Pr}((Y^n,U^n(J),V^n(J,K),Y_s^n,Z^n) \in \mathcal{T}_{\epsilon}^{(n)}) \rightarrow 1$ as $n \rightarrow \infty$, we have that  $H(Y^n|J,K,Y_s^n,Z^n) \geq n[H(Y|U,V,Y_s,Z)-\delta_{\epsilon}]$
\end{lemma}

\textit{Proof of Lemma \ref{lemma:entropy_bound}} We consider the following bound.
	\begin{align*}
	&H(Y^n|J,K,Y_s^n,Z^n) \\
	&= H(Y^n,J,K,Y_s^n,Z^n)-H(J,K)-H(Y_s^n,Z^n|J,K)\\
	&\geq nH(Y,Y_s,Z) + I(X^n;J,K|Y^n,Y_s^n,Z^n)-H(J,K)\\&\qquad -H(Y_s^n,Z^n|J,K)\\
	&= nH(X,Y,Y_s,Z) - H(X^n|J,K,Y^n,Y_s^n,Z^n)-H(J,K)\\&\qquad -H(Y_s^n,Z^n|J,K)\\
	&\overset{(a)}{\geq} n[H(X,Y,Y_s,Z)-H(X|U,V,Y,Y_s,Z)-I(X;U,V)\\ &\qquad -H(Y_s,Z|U,V)-\delta_{\epsilon}]\\
	&\overset{(b)}{=} n[H(Y|U,V,Y_s,Z)-\delta_{\epsilon}],
	\end{align*}
	where $(a)$ follows from the codebook generation where $J \in [1:2^{n(I(X;U)+ \delta_{\epsilon})}]$ and $K \in [1:2^{n(I(X;V|U)+ \delta_{\epsilon})}]$, and from bounding $H(X^n|J,K,Y^n,Y_s^n,Z^n)$ and $H(Y_s^n,Z^n|J,K)$ similarly as in Lemma \ref{lemma:lower_bound}, and $(b)$ follows from the Markov chain $(U,V)-X-(Y,Y_s,Z)$. 
	
The cardinality bounds can be proved using the support lemma \cite{ckITCT11}.

\section{Proof of Theorem \ref{theorem:outer}}\label{app:proof_outer}
Let $(R,D,L)$ be an achievable tuple. Define $U_i \triangleq (W,Z^{i-1},Y_{i+1}^n)$, $V_i \triangleq (W,Z^{i-1},Y^{n\setminus i})$ and $T_i=(W,Y_{i+1}^n,Y_s^{n\setminus i},Z^{n\setminus i})$ which satisfy $(V_i,T_i)-X_i-(Y_i,Y_{p,i},Y_{s,i},Z_i)$ and $U_i-V_i-(T_i,X_i,Y_i,Y_{p,i},Y_{s,i},Z_i)$ for all $i=1,\ldots,n$. From properties of the entropy function we have that 
\begin{align*}
n(R&+\delta_n) \geq H(W)\\
&\geq I(X^n,Z^n;W|Y^n)\\
&= \sum_{i=1}^n H(X_i,Z_i|Y_i) -H(X_i,Z_i|W,X^{i-1},Z^{i-1},Y^n)\\
&\overset{(a)}{\geq} \sum_{i=1}^n H(X_i,Z_i|Y_i) -H(X_i,Z_i|V_i,Y_i)\\
&\geq \sum_{i=1}^n I(X_i;V_i|Y_i),
\end{align*}
where $(a)$ follows from the definition of $V_i$,
and
\begin{align*}
n(D+\delta_n) &\geq nE[d^{(n)}(Y_p^n,g^{(n)}(W,Y^n))]\\
&= \sum_{i=1}^n E[d(Y_{p,i},g^{(n)}_i(W,Y^n))]\\
&\overset{(a)}{=} \sum_{i=1}^n E[d(Y_{p,i},g_i(V_i,Y_i))],
\end{align*}
where $(a)$ follows from the definition of $V_i$ which includes $(W,Y^{n\setminus i})$, implying that  there exists a function $g_i(\cdot)$ such that $g_i(V_i,Y_i)=g^{(n)}_i(W,Y^n)$,
and finally
\begin{align*}
n(L&+\delta_n) \geq I(Y_s^n;W,Z^n)\\
&= H(Y_s^n) - H(Y_s^n,X^n|W,Z^n) +H(X^n|W,Y_s^n,Z^n)\\
&= H(Y_s^n) - H(Y_s^n,X^n|W) +I(Y_s^n,X^n;Z^n|W)\\&\qquad +H(X^n|W,Y_s^n,Z^n)\\
&= H(Y_s^n) - H(Y_s^n,X^n|W,Y^n)-I(Y_s^n,X^n;Y^n|W) \\&\qquad +I(Y_s^n,X^n;Z^n|W)+H(X^n|W,Y_s^n,Z^n)\\
&\overset{(a)}{=} H(Y_s^n) - I(X^n;Y_s^n,Z^n|W,Y^n)-H(Y_s^n|X^n,Y^n)\\&\qquad-I(Y_s^n,X^n;Y^n|W)  +I(Y_s^n,X^n;Z^n|W)\\
&\qquad +I(X^n;Y^n|W,Y_s^n,Z^n)\\
&\overset{(b)}{=} H(Y_s^n) - I(X^n;Y_s^n,Z^n|W,Y^n)-H(Y_s^n|X^n,Y^n)\\&\qquad-I(Z^n;W)+I(Y^n;W)-I(Y_s^n,X^n;Y^n)  \\&\qquad +I(Y_s^n,X^n;Z^n)+I(X^n;Y^n|W,Y_s^n,Z^n),
\end{align*}
where $(a)$ follows from the Markov chain $Y_s^n-(X^n,Y^n)-W$ and $(b)$ follows from the Markov chain $(Y^n,Z^n)-(Y_s^n,X^n)-W$. 

Continuing the chain of inequalities, we get
\begin{align*}
&n(L+\delta_n) \\
& \overset{(c)}{\geq} \sum_{i=1}^n H(Y_{s,i})-H(Y_{s,i},Z_i|V_i,Y_i) + H(Z_i|X_i,Y_i,Y_{s,i})\\&\qquad  -I(Z_i;W|Z^{i-1})+I(Y_i;W|Y_{i+1}^n)-I(Y_{s,i},X_i;Y_i)  \\&\qquad +I(Y_{s,i},X_i;Z_i) + I(X_i;Y_i|W,Y_{i+1}^n,Y_s^n,Z^n)\\  
&\overset{(d)}{=} \sum_{i=1}^n I(Y_{s,i};V_i,Y_i)-I(Z_i;X_i|V_i,Y_i,Y_{s,i}) -I(Y_{s,i},X_i;Y_i) \\&\qquad  +I(Y_{s,i},X_i;Z_i) -I(Z_i;W,Z^{i-1}) +I(Y_i;W,Y_{i+1}^n)\\  
&\qquad +I(X_i;Y_i|T_i,Y_{s,i},Z_i)\\
&\overset{(e)}{=} \sum_{i=1}^n I(Y_{s,i};V_i,Y_i)-I(Z_i;X_i|V_i,Y_i,Y_{s,i})  -I(Y_{s,i},X_i;Y_i)  \\&\qquad  +I(Y_{s,i},X_i;Z_i) -I(Z_i;W,Y_{i+1}^n,Z^{i-1}) \\&\qquad  +I(Y_i;W,Y_{i+1}^n,Z^{i-1})+I(X_i;Y_i|T_i,Y_{s,i},Z_i)\\  
&\overset{(f)}{=} \sum_{i=1}^n I(Y_{s,i};V_i,Y_i)-I(Z_i;X_i|V_i,Y_i,Y_{s,i})  \\&\qquad -I(Y_{s,i},X_i;Y_i|U_i) +I(Y_{s,i},X_i;Z_i|U_i)\\&\qquad +I(X_i;Y_i|T_i,Y_{s,i},Z_i),
\end{align*}
where $(c)$ follows from the Markov chain $(Y_s^n,Z^n)-(X^n,Y^n)-W$, the definition of $V_i$, and the Markov chain $Y_i-(W,X_i,Y_{i+1}^n,Y_s^n,Z^n)-X^{n \setminus i}$, $(d)$ follows from  the Markov chain $Z_i-(X_i,Y_i,Y_{s,i})-V_i$, the definition of  $T_i$, and the facts that $Z_i$ is independent of $Z^{i-1}$ and $Y_i$ is independent of $Y_{i+1}^n$, $(e)$ follows from the Csiszar's sum identity \cite{ekNIT11}, and $(f)$ follows from the definition of $U_i$ and the Markov chain $(Y_i,Z_i)-(X_i,Y_{s,i})-U_i$.

We ends the proof by the standard time-sharing argument and letting $n \rightarrow \infty$ and $\delta_n \rightarrow 0$. 

\section{Proof of Gaussian with Quadratic Distortion Example}\label{app:proof_Gaussian_example}
Recalling that $X \sim \mathcal{N}(0,N_x)$, $Y_s= X + \tilde{N}_s$, where $\tilde{N}_s \sim \mathcal{N}(0,N_s) \perp X$, and $Y_p= Y_s + \tilde{N}_p$, where $\tilde{N}_p \sim \mathcal{N}(0,N_p-N_s) \perp Y_s$. 

\emph{Proof of Achievability}:  Let us choose $V=X+Q, Q\sim \mathcal{N}(0,N_q)$ independent of $X$, and choose $g(V)$ to be an MMSE estimate of $Y_p$ given $V$.
With these choices of $V$ and $g(\cdot)$, it can be shown that
\begin{align*}
I(X;V) &= h(V)-h(V|X)\\
&= \frac{1}{2}\log\bigg(\frac{N_x+N_q}{N_q}\bigg),
\end{align*}
and
\begin{align*}
I(Y_s;V)   &= h(Y_s)-h(Y_s|V)\\
&= \frac{1}{2}\log\bigg(\frac{N_x+N_s}{N_s + \frac{N_xN_q}{N_x+N_q}}\bigg),
\end{align*}
and lastly
\begin{align*}
E[d(Y_p,g(V))]&=E[(Y_p-g(V))^2]\\
&=N_p+ \frac{N_xN_q}{N_x+N_q}.
\end{align*}

Letting $D=E[d(Y_p,g(V))]=N_p+ \frac{N_xN_q}{N_x+N_q}$ and substituting it into the constraints on $R$ and $L$ complete the achievablity part.

\emph{Proof of Converse}: From the problem formulation, the joint PMF $P_{X^n,Y_p^n,Y_s^n,W,\hat{Y}_p^n}$ is given by
\begin{align*}
P_{X^n,Y_s^n}P_{Y_p^n|Y_s^n} P_{W|X^n}1_{\{\hat{Y}_p^n = g^{(n)}(W)\}}.
\end{align*}
From the EPI, we have that 
\begin{equation}
2^{\frac{2}{n}h(Y_s^n|W)} \geq 2^{\frac{2}{n}h(X^n|W)} +2^{\frac{2}{n}h(\tilde{N}_s^n)}\label{eq:EPI1}
\end{equation}
and
\begin{align}
2^{\frac{2}{n}h(Y_p^n|W)} &\geq 2^{\frac{2}{n}h(Y_s^n|W)} +2^{\frac{2}{n}h(\tilde{N}_p^n)}\label{eq:EPI2}\\
&\geq 2^{\frac{2}{n}h(X^n|W)} +2^{\frac{2}{n}h(\tilde{N}_s^n)}+2^{\frac{2}{n}h(\tilde{N}_p^n)},\label{eq:EPI3}
\end{align}
where the last inequality follows from \eqref{eq:EPI1}.
Then it follows that
\begin{align}
n(R + \delta_n) &\geq H(W )\nonumber\\
&\geq I(X^n;W)\nonumber\\
&= h(X^n)-h(X^n|W)\nonumber\\
&\geq h(X^n) \nonumber\\&\qquad - \frac{n}{2}\log(2^{\frac{2}{n}h(Y_p^n|W)} -2^{\frac{2}{n}h(\tilde{N}_s^n)}-2^{\frac{2}{n}h(\tilde{N}_p^n)}) ,\label{eq:R_converse}
\end{align}
where the last inequality follows from \eqref{eq:EPI3}.

\begin{align}
n(L + \delta_n) &\geq  I(Y_s^n;W)\nonumber\\
&= h(Y_s^n)-h(Y_s^n|W)\nonumber\\
&\geq h(Y_s^n)- \frac{n}{2}\log(2^{\frac{2}{n}h(Y_p^n|W)} -2^{\frac{2}{n}h(\tilde{N}_p^n)}), \label{eq:L_converse}
\end{align}
where the last inequality follows from \eqref{eq:EPI2}.

\begin{align}
D + \delta_n &\geq E[d^{(n)}(Y_p^n,g^{(n)}(W))]\nonumber \\
&= \frac{1}{n} \sum_{i=1}^n E[(Y_{p,i}-g_i^{(n)}(W))^2]\nonumber\\
&= \frac{1}{2\pi e}2^{\log(\frac{2 \pi e}{n} \sum_{i=1}^n E[(Y_{p,i}-g_i^{(n)}(W))^2])}\nonumber\\
&\overset{(a)}{\geq} \frac{1}{2\pi e}2^{\frac{1}{n}\sum_{i=1}^n\log(  2 \pi eE[(Y_{p,i}-g_i^{(n)}(W))^2])}\nonumber\\
&\overset{(b)}{\geq}  \frac{1}{2\pi e}2^{\frac{1}{n}\sum_{i=1}^n\log(  2 \pi e E[\text{var}(Y_{p,i}|W)])}\nonumber\\
&\geq  \frac{1}{2\pi e}2^{\frac{2}{n}\sum_{i=1}^n h(Y_{p,i}|W)}\nonumber\\
&\geq \frac{1}{2\pi e} 2^{\frac{2}{n} h(Y_p^n|W)}, \label{eq:D_converse}
\end{align}
where $(a)$ follows from Jensen's inequality \cite{ekNIT11} and the fact that $\log(\cdot)$ is a concave function, and $(b)$ follows from the fact that $E[\text{var}(Y_{p,i}|W)]$ is the MMSE over all possible estimator of $Y_{p,i}$ for each $i=1,\ldots,n$.

Combining \eqref{eq:D_converse} with  \eqref{eq:R_converse} and \eqref{eq:L_converse}, and letting $n \rightarrow \infty$ and $\delta_n \rightarrow 0$ complete the converse part.

\section{Proof of Secure Information Bottleneck Example: binary source}\label{app:proof_secureIB_binary_example}
Recalling that $X \sim \text{Bernoulli}(1/2)$, $Y_s$ is an output of a BSC($p$) with input $X$, and $Y_p$ is an output of a BSC($q$) with input $Y_s$.

\emph{Proof of Achievability}:  Let $V$ be an output of a BSC with input $X$. We have that 
\begin{align}
H(Y_p|V)&=h(h^{-1}(H(X|V))*p*q)\label{eq:secureIB_achive1}\\
&=h(h^{-1}(H(Y_s|V))*q).\label{eq:secureIB_achive2}
\end{align}
We let $D'=H(Y_p)-H(Y_p|V)$. Combining \eqref{eq:secureIB_achive1} with the constraint $R \geq H(X)-H(X|V)$ and noting that $u*q$ is an increasing function in $u \in [0,1/2]$ for some fixed $q \in [0,1/2]$, we get  $D' \leq H(Y_p)-h(h^{-1}(H(X)-R)*p*q))$. Similarly, combining \eqref{eq:secureIB_achive2} with $L \geq H(Y_s)-H(Y_s|V)$ gives $D' \leq H(Y_p)-h(h^{-1}(H(Y_s)-L)*q))$.

\emph{Proof of Converse}: From Mrs. Gerber's lemma, we have
\begin{align}
H(Y_s|V) &\geq h(h^{-1}(H(X|V))*p)\label{eq:mrsGerber1}\\
H(Y_p|V) &\geq h(h^{-1}(H(Y_s|V))*q).\label{eq:mrsGerber2}
\end{align} 
Since $h(h^{-1}(u)*q)$ is an increasing function in $u$ for $q \in [0,1/2]$, combining \eqref{eq:mrsGerber1} and \eqref{eq:mrsGerber2} gives
\begin{equation}
H(Y_p|V) \geq h(h^{-1}(H(X|V))*p*q). \label{eq:mrsGerber3}
\end{equation}

Then we have that 
\begin{align*}
D' &\leq I(Y_p;V)\\
&\overset{(a)}{\leq} H(Y_p)-h(h^{-1}(H(Y_s|V))*q)\\
&\overset{(b)}{\leq} H(Y_p)- h(h^{-1}(H(Y_s)-L)*q),
\end{align*}
where $(a)$ follows from \eqref{eq:mrsGerber2} and $(b)$ follows from the facts that $h(h^{-1}(H(Y_s|V))*q)$ is an increasing function in $H(Y_s|V)$ and that $L \geq H(Y_s)-H(Y_s|V)$.

Similarly, we have
\begin{align*}
D' &\leq I(Y_p;V)\\
&\overset{(a)}{\leq} H(Y_p)- h(h^{-1}(H(X|V))*p*q)\\
&\overset{(b)}{\leq} H(Y_p)- h(h^{-1}(H(X)-R)*p*q),
\end{align*}
where $(a)$ follows from \eqref{eq:mrsGerber3} and $(b)$ follows from the facts that $h(h^{-1}(H(X|V))*p*q)$ is an increasing function in $H(X|V)$ and that $R \geq H(X)-H(X|V)$.

\bibliographystyle{IEEEtran}
\bibliography{IEEEabrv,secure_info_bottleneck}

\begin{thebibliography}{10}
\providecommand{\url}[1]{#1}
\csname url@samestyle\endcsname
\providecommand{\newblock}{\relax}
\providecommand{\bibinfo}[2]{#2}
\providecommand{\BIBentrySTDinterwordspacing}{\spaceskip=0pt\relax}
\providecommand{\BIBentryALTinterwordstretchfactor}{4}
\providecommand{\BIBentryALTinterwordspacing}{\spaceskip=\fontdimen2\font plus
\BIBentryALTinterwordstretchfactor\fontdimen3\font minus
  \fontdimen4\font\relax}
\providecommand{\BIBforeignlanguage}[2]{{%
\expandafter\ifx\csname l@#1\endcsname\relax
\typeout{** WARNING: IEEEtran.bst: No hyphenation pattern has been}%
\typeout{** loaded for the language `#1'. Using the pattern for}%
\typeout{** the default language instead.}%
\else
\language=\csname l@#1\endcsname
\fi
#2}}
\providecommand{\BIBdecl}{\relax}
\BIBdecl

\bibitem{srpUPTI13}
L.~Sankar, S.~R. Rajagopalan, and H.~V. Poor, ``Utility-privacy tradeoffs in
  databases: An information-theoretic approach,'' \emph{IEEE Trans. Inf.
  Forensic Secur.}, vol.~8, no.~6, pp. 838--852, June 2013.

\bibitem{dfPASI12}
F.~du~Pin~Calmon and N.~Fawaz, ``Privacy against statistical inference,'' in
  \emph{Proc. Allerton Conf. Commun. Control Comput}, Oct 2012.

\bibitem{yASCP83}
H.~Yamamoto, ``A source coding problem for sources with additional outputs to
  keep secret from the receiver or wiretappers (corresp.),'' \emph{{IEEE}
  Trans. Inf. Theory}, vol.~29, no.~6, pp. 918--923, Nov 1983.

\bibitem{yCTFS94}
------, ``Coding theorems for shannon's cipher system with correlated source
  outputs, and common information,'' \emph{{IEEE} Trans. Inf. Theory}, vol.~40,
  no.~1, pp. 85--95, Jan 1994.

\bibitem{prOSDS07}
V.~Prabhakaran and K.~Ramchandran, ``On secure distributed source coding,''
  \emph{Proc. IEEE Inf. Theory Workshop}, pp. 442--447, 2007.

\bibitem{gepLCWS08}
D.~G\"{u}nd\"{u}z, E.~Erkip, and H.~V. Poor, ``Lossless compression with
  security constraints,'' in \emph{Proc. IEEE ISIT}, 2008, pp. 111--115.

\bibitem{vpSMSC13}
J.~Villard and P.~Piantanida, ``Secure multiterminal source coding with side
  information at the eavesdropper,'' \emph{{IEEE} Trans. Inf. Theory}, vol.~59,
  no.~6, pp. 3668--3692, June 2013.

\bibitem{scRDTS14}
C.~Schieler and P.~Cuff, ``Rate-distortion theory for secrecy systems,''
  \emph{{IEEE} Trans. Inf. Theory}, vol.~60, no.~12, pp. 7584--7605, Dec 2014.

\bibitem{mfUP98}
N.~Merhav and M.~Feder, ``Universal prediction,'' \emph{{IEEE} Trans. Inf.
  Theory}, vol.~44, no.~6, pp. 2124--2147, Oct 1998.

\bibitem{cwMSCU14}
T.~Courtade and T.~Weissman, ``Multiterminal source coding under logarithmic
  loss,'' \emph{{IEEE} Trans. Inf. Theory}, vol.~60, no.~1, pp. 740--761, Jan
  2014.

\bibitem{msfmFIBT14}
A.~Makhdoumi, S.~Salamatian, N.~Fawaz, and M.~Medard, ``From the information
  bottleneck to the privacy funnel,'' in \emph{IEEE Information Theory Workshop
  (ITW), 2014}, Nov 2014, pp. 501--505.

\bibitem{tpbTIBM00}
N.~{Tishby}, F.~C. {Pereira}, and W.~{Bialek}, ``{The information bottleneck
  method},'' in \emph{Proc. Allerton Conf. Commun. Control Comput}, 1999.

\bibitem{tspDLSC13}
R.~Tandon, L.~Sankar, and H.~V. Poor, ``Discriminatory lossy source coding:
  Side information privacy,'' \emph{{IEEE} Trans. Inf. Theory}, vol.~59, no.~9,
  pp. 5665--5677, Sept 2013.

\bibitem{hbRDSA85}
C.~Heegard and T.~Berger, ``Rate distortion when side information may be
  absent,'' \emph{{IEEE} Trans. Inf. Theory}, vol.~31, pp. 727--734, Nov 1985.

\bibitem{msITPI15}
B.~Moraffah and L.~Sankar, ``Information-theoretic private interactive
  mechanism,'' in \emph{Proc. Allerton Conf. Commun. Control Comput}, 2015.

\bibitem{ekNIT11}
A.~{El Gamal} and Y.~H. Kim, \emph{Network Information Theory}.\hskip 1em plus
  0.5em minus 0.4em\relax Cambridge University Press, 2011.

\bibitem{wzTRDF76}
A.~D. Wyner and J.~Ziv, ``The rate-distortion function for source coding with
  side information at the decoder,'' \emph{{IEEE} Trans. Inf. Theory}, vol.~22,
  no.~1, pp. 1--10, Jan 1976.

\bibitem{aalLSSC15}
S.~{Asoodeh}, F.~{Alajaji}, and T.~{Linder}, ``Lossless secure source coding:
  {Y}amamoto's setting,'' in \emph{Proc. Allerton Conf. Commun. Control
  Comput}, 2015.

\bibitem{wTRDF78}
A.~D. Wyner, ``The rate-distortion function for source coding with side
  information at the decoderÑ-part {II}: General sources,'' \emph{Inf.
  Control}, no.~38, pp. 60Ð--80, 1978.

\bibitem{bPRML06}
C.~M. Bishop, \emph{Pattern Recognition and Machine Learning (Information
  Science and Statistics)}.\hskip 1em plus 0.5em minus 0.4em\relax Secaucus,
  NJ, USA: Springer-Verlag New York, Inc., 2006.

\bibitem{aabgEBDM07}
T.~Andre, M.~Antonini, M.~Barlaud, and R.~Gray, ``Entropy-based distortion
  measure and bit allocation for wavelet image compression,'' \emph{{IEEE}
  Trans. Image Process.}, vol.~16, no.~12, pp. 3058--3064, Dec 2007.

\bibitem{ckOSSC13}
Y.~Chia and K.~Kittichokechai, ``On secure source coding with side information
  at the encoder,'' \emph{CoRR}, 2013, abs/1307.0974.

\bibitem{ctERSW02}
G.~Chechik and N.~Tishby, ``Extracting relevant structures with side
  information,'' Advances in Neural Information Processing Systems 15, 2002.

\bibitem{stAIB99}
N.~Slonim and N.~Tishby, ``Agglomerative information bottleneck.''\hskip 1em
  plus 0.5em minus 0.4em\relax MIT Press, 1999, pp. 617--623.

\bibitem{wmRIOG14}
A.~Winkelbauer and G.~Matz, ``Rate-information-optimal gaussian channel output
  compression,'' in \emph{48th Annual Conference on Information Sciences and
  Systems (CISS)}, March 2014.

\bibitem{ckITCT11}
I.~Csisz\'{a}r and J.~K\"{o}rner, \emph{Information Theory: Coding Theorems for
  Discrete Memoryless Systems}.\hskip 1em plus 0.5em minus 0.4em\relax Cambrige
  University Press, 2011.

\end{thebibliography}
\end{document}